\documentclass[11pt]{article}%
\usepackage{amsfonts}
\usepackage{amsmath}
\usepackage{amssymb}
\usepackage{graphicx}
\usepackage{fullpage}
\usepackage[colorlinks=true,linkcolor=blue,citecolor=red,plainpages=false,pdfpagelabels]%
{hyperref}%
\providecommand{\U}[1]{\protect\rule{.1in}{.1in}}
\newtheorem{theorem}{Theorem}

\newtheorem{remark}[theorem]{Remark}

\newenvironment{proof}[1][Proof]{\noindent\textbf{#1.} }{\ \rule{0.5em}{0.5em}}
\begin{document}

\title{\textbf{Hadamard quantum broadcast channels}}
\author{Qingle Wang\thanks{State Key Laboratory of Networking and Switching
Technology, Beijing University of Posts and Telecommunications, Beijing
100876, China} \thanks{Hearne Institute for Theoretical Physics, Department of
Physics and Astronomy, Louisiana State University, Baton Rouge, Louisiana
70803, USA}
\and Siddhartha Das\footnotemark[2]
\and Mark M. Wilde\footnotemark[2] \thanks{Center for Computation and Technology
Louisiana State University, Baton Rouge, Louisiana 70803, USA}}
\date{\today}
\maketitle

\begin{abstract}
We consider three different communication tasks for quantum broadcast
channels, and we determine the capacity region of a Hadamard broadcast channel
for these various tasks. We define a Hadamard broadcast channel to be such
that the channel from the sender to one of the receivers is
entanglement-breaking and the channel from the sender to the other receiver is
complementary to this one. As such, this channel is a quantum generalization
of a degraded broadcast channel, which is well known in classical information
theory. The first communication task we consider is classical communication to
both receivers, the second is quantum communication to the stronger receiver
and classical communication to other, and the third is entanglement-assisted
classical communication to the stronger receiver and unassisted classical
communication to the other. The structure of a Hadamard broadcast channel
plays a critical role in our analysis:\ the channel to the weaker receiver can
be simulated by performing a measurement channel on the stronger receiver's
system, followed by a preparation channel. As such, we can incorporate the
classical output of the measurement channel as an auxiliary variable and solve
all three of the above capacities for Hadamard broadcast channels, in this way
avoiding known difficulties associated with quantum auxiliary variables.

\end{abstract}





\section{Introduction}

Broadcast channels model the communication of a single sender to multiple
receivers \cite{C72}. They have been explored extensively in classical
information theory \cite{C72,M79,cover1998comments,el2010lecture}, with a
variety of coding schemes known, including the superposition coding method
\cite{C72}. The capacity of a classical broadcast channel has been solved in
certain cases \cite{el2010lecture}\ but remains unsolved in the general case,
being a well known open problem in network classical information theory.

Quantum broadcast channels were introduced in \cite{YHD2006}\ and take on a
particular relevance in quantum information theory, due to the no-cloning
theorem \cite{nat1982,D82}\ and associated \textquotedblleft
no-go\textquotedblright\ results \cite{BCFJS96,BBLW07,PHH08,Piani16}. A
variety of information-theoretic results are now known for quantum broadcast
channels. Refs.~\cite{YHD2006,SW11}\ established a quantum generalization of
the superposition coding method for sending classical information over a
broadcast channel. Ref.~\cite{YHD2006}\ established a method for sending
classical information to one receiver while sending quantum information to the
other. Other information-theoretic results having to do with a variety of
communication tasks for quantum broadcast channels are available in
\cite{DHL10,RSW14,HM15,STW16,TSW16}.

In this paper, we determine the classical capacity region, the
classical--quantum capacity region, and the partially entanglement-assisted
classical capacity region of Hadamard quantum broadcast channels. That is, we
determine the optimal rates at which a sender can transmit classical
information to two receivers, the optimal rates at which a sender can
communicate quantum information to one receiver and classical information to
the other receiver, as well as the optimal rates at which a sender can
communicate classical information to both receivers while sharing entanglement
with one of them, whenever the underlying channel is a Hadamard broadcast
channel. Hadamard broadcast channels are such that the sender Alice's input is
isometrically embedded in the Hilbert space of two receivers, Bob and Charlie,
the channel from Alice to Charlie is entanglement-breaking \cite{HSR03},\ and
the channel from Alice to Bob is complementary to the aforementioned one. The
channel from Alice to Bob is known as a Hadamard channel \cite{KMNR07}, and
for this reason, we call the corresponding broadcast channel a Hadamard
broadcast channel. Single-sender single-receiver Hadamard channels have a
complete characterization in terms of all of their capacities
\cite{BHTW10,WH10b,W15book}, and so our results here represent a further
exploration of these ideas in the domain of broadcast channels. An interesting
example of a Hadamard channel is a quantum-limited amplifier channel
\cite{Holevo2008,giovannetti2014ultimate}, which has appeared in a variety of
contexts in quantum information theory due to its connections with approximate
cloning \cite{SIGA05,B09,Fan2014241}.

The rest of the paper proceeds as follows. In the next section, we provide a
definition of a Hadamard quantum broadcast channel. Section~\ref{sec:CC-task}
details our first main result about the classical capacity of a Hadamard
broadcast channel. Therein we define the communication task, we review the
rate region achievable when using the superposition coding method from
\cite{YHD2006,SW11}, and we detail a proof that the region is single-letter
for Hadamard broadcast channels. Section~\ref{sec:CQ-task} details our second
main result about the classical--quantum capacity of a Hadamard broadcast
channel. Therein we define the communication task, review the achievable rate
region from \cite{YHD2006}, and thereafter give the converse proof for the
classical--quantum capacity region. Section~\ref{sec:EAC-task}\ gives our
third main result about the partially entanglement-assisted classical capacity
of a Hadamard broadcast channel. In Section~\ref{sec:conclusion}, we conclude
with a brief summary and some open questions for future research. We point the
reader to \cite{W15book} for basics of quantum information theory and for
background on the standard notation and concepts being used in our paper.

\section{Hadamard quantum broadcast channels}

We define a quantum broadcast channel $\mathcal{N}_{A\rightarrow
BC}^{\operatorname{H}}$\ to be Hadamard if it has the following action on an
input state $\sigma_{A}$:%
\begin{equation}
\mathcal{N}_{A\rightarrow BC}^{\operatorname{H}}(\sigma_{A})\equiv\sum
_{x,y}\langle\phi^{x}|_{A}\sigma_{A}|\phi^{y}\rangle_{A}|x\rangle\langle
y|_{B}\otimes|\psi^{x}\rangle\langle\psi^{y}|_{C}%
,\label{eq:hadamard-broadcast-def}%
\end{equation}
where the vectors $\{|\phi^{x}\rangle_{A}\}_{x}$ are such that they form a
positive operator-valued measure (POVM) $\sum_{x}|\phi^{x}\rangle\langle
\phi^{x}|_{A}=I_{A}$, $\{|x\rangle_{B}\}_{x}$ is an orthonormal basis, and
$\{|\psi^{x}\rangle_{C}\}_{x}$ is a set of states. Note that the channel in
\eqref{eq:hadamard-broadcast-def} is an isometric channel \cite[Section~4.6.3]%
{W15book}, meaning that its can be reversed. This is a key fact that we
exploit in our paper. The reduced channel to Bob is a Hadamard channel
\cite{KMNR07} of the following form:%
\begin{align}
\mathcal{N}_{A\rightarrow B}^{\operatorname{H}}(\sigma_{A})  &  \equiv
(\operatorname{Tr}_{C}\circ\mathcal{N}_{A\rightarrow BC}^{\operatorname{H}%
})(\sigma_{A})\\
&  =\sum_{x,y}\langle\phi^{x}|_{A}\sigma_{A}|\phi^{y}\rangle_{A}\langle
\psi^{y}|\psi^{x}\rangle_{C}|x\rangle\langle y|_{B}.
\end{align}
The reduced channel to Charlie is an entanglement-breaking channel
\cite{HSR03} of the following form:%
\begin{align}
\mathcal{N}_{A\rightarrow C}^{\operatorname{H}}(\sigma_{A})  &  \equiv
(\operatorname{Tr}_{B}\circ\mathcal{N}_{A\rightarrow BC}^{\operatorname{H}%
})(\sigma_{A})\\
&  =\sum_{x}\langle\phi^{x}|_{A}\sigma_{A}|\phi^{x}\rangle_{A}|\psi^{x}%
\rangle\langle\psi^{x}|_{C},
\end{align}
meaning that the action of the channel is to measure the input with respect to
the POVM $\{|\phi^{x}\rangle\langle\phi^{x}|_{A}\}_{x}$ and then prepare the
state $|\psi^{x}\rangle_{C}$ at the output if the measurement outcome is $x$.
Such channels have the property that Bob can apply the following isometry to
his output received from the channel:%
\begin{equation}
V_{B\rightarrow BYC^{\prime}}\equiv\sum_{x}|x\rangle_{B}\langle x|_{B}%
\otimes|x\rangle_{Y}\otimes|\psi^{x}\rangle_{C^{\prime}},
\end{equation}
where system $C^{\prime}$ is isomorphic to system $C$, and if he then discards
the $B$ and $Y$ systems, the effect is to simulate the channel to Charlie.
That is,%
\begin{equation}
\operatorname{Tr}_{BY}\circ\mathcal{V}_{B\rightarrow BYC^{\prime}}%
\circ\mathcal{N}_{A\rightarrow B}^{\operatorname{H}}=\mathcal{N}_{A\rightarrow
C}^{\operatorname{H}}.
\end{equation}
The channel $\mathcal{D}_{B\rightarrow C^{\prime}}\equiv\operatorname{Tr}%
_{BY}\circ\mathcal{V}_{B\rightarrow BYC^{\prime}}$ is known as the degrading
channel \cite{cmp2005dev}. We can also consider the degrading channel as
arising in two steps:\ a measurement channel $\mathcal{M}_{B\rightarrow
Y}(\cdot)=\sum_{x}|x\rangle\langle x|_{Y}(\cdot)|x\rangle\langle x|_{Y}$
followed by a preparation channel $\mathcal{P}_{Y\rightarrow C^{\prime}}%
(\cdot)=\sum_{x}\langle x|(\cdot)|x\rangle_{Y}|\psi^{x}\rangle\langle\psi
^{x}|_{C^{\prime}}$, so that%
\begin{equation}
\mathcal{D}_{B\rightarrow C^{\prime}}=\mathcal{P}_{Y\rightarrow C^{\prime}%
}\circ\mathcal{M}_{B\rightarrow Y}.
\end{equation}

\section{Classical capacity region of a Hadamard broadcast
channel\label{sec:CC-task}}

\subsection{Definition of the classical capacity region of a broadcast
channel}

We begin by defining the classical capacity region of a quantum broadcast
channel \cite{YHD2006}. Let $\mathcal{N}_{A\rightarrow BC}$ denote a quantum
broadcast channel from a sender Alice to receivers Bob and Charlie. Let
$n\in\mathbb{N}$, $M_{B},M_{C}\in\mathbb{N}$, and $\varepsilon\in\lbrack0,1]$.
An $(n,M_{B},M_{C},\varepsilon)$ code for classical communication over the
broadcast channel $\mathcal{N}_{A\rightarrow BC}$ consists of quantum
codewords $\left\{  \rho_{A^{n}}^{m_{1},m_{2}}\right\}  _{m_{1},m_{2}}$ and
POVMs $\{\Lambda_{B^{n}}^{m_{1}}\}_{m_{1}}$ and $\{\Gamma_{C^{n}}^{m_{2}%
}\}_{m_{2}}$ such that the message pair $(m_{1},m_{2})$ is communicated with
average success probability not smaller than $1-\varepsilon$:%
\begin{equation}
\frac{1}{M_{B}M_{C}}\sum_{m_{1},m_{2}}\operatorname{Tr}\{(\Lambda_{B^{n}%
}^{m_{1}}\otimes\Gamma_{C^{n}}^{m_{2}})\mathcal{N}_{A\rightarrow BC}^{\otimes
n}(\rho_{A^{n}}^{m_{1},m_{2}})\}\geq1-\varepsilon.
\label{eq:success-prob-broadcast}%
\end{equation}
Note that, in the above, $m_{1} \in\{1, \ldots, M_{1}\}$ and $m_{2} \in\{1,
\ldots, M_{2}\}$.

A rate pair $(R_{B},R_{C})$ is achievable for classical communication on
$\mathcal{N}_{A\rightarrow BC}$\ if for all $\varepsilon\in(0,1)$, $\delta>0$,
and sufficiently large $n$, there exists an $(n,2^{n\left[  R_{B}%
-\delta\right]  },2^{n\left[  R_{C}-\delta\right]  },\varepsilon)$ code of the
above form. The classical capacity region of $\mathcal{N}_{A\rightarrow BC}$
is equal to the closure of all achievable rate pairs.

\subsection{Achievable rate region for an arbitrary quantum broadcast channel}

From the superposition coding result in \cite{YHD2006,SW11}, we know the
following achievability statement:

\begin{theorem}
[\cite{YHD2006,SW11}]\label{thm:CC-achieve}Given a quantum broadcast channel
$\mathcal{N}_{A\rightarrow BC}$, a rate pair $(R_{B},R_{C})$ is achievable for
classical communication on $\mathcal{N}_{A\rightarrow BC}$ if%
\begin{align}
R_{B}  &  \leq I(Z;B|W)_{\theta},\label{eq:super-code-1}\\
R_{C}  &  \leq I(W;C)_{\theta},\label{eq:super-code-2}\\
R_{B}+R_{C}  &  \leq I(Z;B)_{\theta}, \label{eq:super-code-3}%
\end{align}
where the information quantities are evaluated with respect to a state
$\theta_{WZBC}$ of the following form:%
\begin{equation}
\sum_{w,z}p_{WZ}(w,z)|w\rangle\langle w|_{W}\otimes|z\rangle\langle
z|_{Z}\otimes\mathcal{N}_{A\rightarrow BC}(\sigma_{A}^{z}),
\end{equation}
with $p_{WZ}$ a probability distribution and $\{\sigma_{A}^{z}\}_{z}$ a set of states.
\end{theorem}

For a Hadamard quantum broadcast channel, the region above simplifies due to
the structure of the channel. That is, the bound in \eqref{eq:super-code-3} is
unnecessary (redundant) for a Hadamard quantum broadcast channel. To see this,
consider that the sum of \eqref{eq:super-code-1}--\eqref{eq:super-code-2}
leads to%
\begin{align}
R_{B}+R_{C}  &  \leq I(Z;B|W)_{\theta}+I(W;C)_{\theta}\\
&  \leq I(Z;B|W)_{\theta}+I(W;B)_{\theta}\\
&  =I(WZ;B)_{\theta}\\
&  =I(Z;B)_{\theta}+I(W;B|Z)_{\theta}\\
&  =I(Z;B)_{\theta}.
\end{align}
The second inequality follows from the data-processing inequality for mutual
information and from the fact that there is a degrading channel taking system
$B$ to system $C$ for a Hadamard broadcast channel. The next two equalities
follow from the chain rule for mutual information. The last equality follows
because the state of systems $W$ and $B$ are product when conditioned on the
value in system $Z$. So then the achievable region for a Hadamard broadcast
channel consists of \eqref{eq:super-code-1}--\eqref{eq:super-code-2}.

\subsection{Classical capacity of a Hadamard broadcast channel}

We now show that the region specified by
\eqref{eq:super-code-1}--\eqref{eq:super-code-2} is in fact the classical
capacity region of a Hadamard broadcast channel (i.e., one can never achieve a
rate outside of this region). That is, we establish the following capacity
theorem for a Hadamard broadcast channel:

\begin{theorem}
\label{thm:CC-HBC} The classical capacity region of a Hadamard broadcast
channel $\mathcal{N}_{A\rightarrow BC}^{\operatorname{H}}$ is the set of rate
pairs $(R_{B},R_{C})$ such that%
\begin{align}
R_{B}  &  \leq I(Z;B|W)_{\theta},\label{eq:HB-ineq-1}\\
R_{C}  &  \leq I(W;C)_{\theta}, \label{eq:HB-ineq-2}%
\end{align}
for some state%
\begin{equation}
\theta_{WZA}=\sum_{w,z}p_{WZ}(w,z)|w\rangle\langle w|_{W}\otimes
|z\rangle\langle z|_{Z}\otimes\varphi_{A}^{z}, \label{eq:code-state}%
\end{equation}
where $p_{WZ}$ is a probability distribution, each $\varphi_{A}^{z}$ is a pure
state, and the information quantities are evaluated with respect to the state
$\theta_{WZBC}=\mathcal{N}_{A\rightarrow BC}^{\operatorname{H}}(\theta_{WZA})$.
\end{theorem}

\begin{proof}
The achievability part follows from a direct application of
Theorem~\ref{thm:CC-achieve}.

For the converse, consider an arbitrary $(n,M_{B},M_{C},\varepsilon)$ code for
the broadcast Hadamard channel $\mathcal{N}_{A\rightarrow BC}%
^{\operatorname{H}}$. Let $\omega_{M_{1}M_{2}B^{n}C^{n}}$ denote the following
state:%
\begin{equation}
\omega_{M_{1}M_{2}B^{n}C^{n}}\equiv\frac{1}{M_{B}M_{C}}\sum_{m_{1},m_{2}%
}|m_{1}\rangle\langle m_{1}|_{M_{1}}\otimes|m_{2}\rangle\langle m_{2}|_{M_{2}}
\otimes\mathcal{N}_{A\rightarrow BC}^{\operatorname{H}\otimes n}(\rho_{A^{n}%
}^{m_{1},m_{2}}),
\end{equation}
so that this is the state before the receivers act with their measurements.
The post-measurement state is as follows:%
\begin{multline}
\omega_{M_{1}M_{2}M_{1}^{\prime}M_{2}^{\prime}}\equiv\sum_{m_{1},m_{1}%
^{\prime},m_{2},m_{2}^{\prime}}p(m_{1},m_{1}^{\prime},m_{2},m_{2}^{\prime
})|m_{1}\rangle\langle m_{1}|_{M_{1}}\otimes|m_{1}^{\prime}\rangle\langle
m_{1}^{\prime}|_{M_{1}^{\prime}}\\
\otimes|m_{2}\rangle\langle m_{2}|_{M_{2}}\otimes|m_{2}^{\prime}\rangle\langle
m_{2}^{\prime}|_{M_{2}^{\prime}},
\end{multline}
where%
\begin{equation}
p(m_{1},m_{1}^{\prime},m_{2},m_{2}^{\prime})\equiv\frac{\operatorname{Tr}%
\{(\Lambda_{B^{n}}^{m_{1}^{\prime}}\otimes\Gamma_{C^{n}}^{m_{2}^{\prime}%
})\mathcal{N}_{A\rightarrow BC}^{\operatorname{H}\otimes n}(\rho_{A^{n}%
}^{m_{1},m_{2}})\}}{M_{B}M_{C}}.
\end{equation}
From the condition in \eqref{eq:success-prob-broadcast}, it follows that%
\begin{equation}
\frac{1}{2}\left\Vert \omega_{M_{1}M_{2}M_{1}^{\prime}M_{2}^{\prime}%
}-\overline{\Phi}_{M_{1}M_{1}^{\prime}}\otimes\overline{\Phi}_{M_{2}%
M_{2}^{\prime}}\right\Vert _{1}\leq\varepsilon,
\end{equation}
where $\overline{\Phi}_{M_{i}M_{i}^{\prime}}\equiv\frac{1}{|M_{i}|}\sum
_{m_{i}}|m_{i}\rangle\langle m_{i}|_{M_{i}}\otimes|m_{i}\rangle\langle
m_{i}|_{M_{i}^{\prime}}$ is the maximally correlated state for $i\in\{1,2\}$.

The reduced state for Charlie can be simulated by acting with the measurement
channel $\mathcal{M}_{B\rightarrow Y}^{\otimes n}$ followed by the preparation
channel $\mathcal{P}_{Y\rightarrow C}^{\otimes n}$, given our assumption of a
Hadamard broadcast channel. That is, we have that%
\begin{equation}
\omega_{M_{1}M_{2}C^{n}}=\mathcal{P}_{Y\rightarrow C}^{\otimes n}(\xi
_{M_{1}M_{2}Y^{n}}),
\end{equation}
where%
\begin{equation}
\xi_{M_{1}M_{2}Y^{n}}\equiv\mathcal{M}_{B\rightarrow Y}^{\otimes n}%
(\omega_{M_{1}M_{2}B^{n}}).
\end{equation}

Let a spectral decomposition for the state $\rho_{A^{n}}^{m_{1},m_{2}}$ be as
follows:%
\begin{equation}
\rho_{A^{n}}^{m_{1},m_{2}}=\sum_{t}p(t|m_{1},m_{2})\varsigma_{A^{n}}%
^{m_{1},m_{2},t},
\end{equation}
with each $\varsigma_{A^{n}}^{m_{1},m_{2},t}$ pure, so that the following
state $\omega_{M_{1}M_{2}B^{n}C^{n}T}$ is an extension of $\omega_{M_{1}%
M_{2}B^{n}C^{n}}$:%
\begin{multline}
\omega_{M_{1}M_{2}B^{n}C^{n}T}\equiv\frac{1}{M_{B}M_{C}}\sum_{m_{1},m_{2}%
,t}|m_{1}\rangle\langle m_{1}|_{M_{1}}\otimes|m_{2}\rangle\langle
m_{2}|_{M_{2}}\\
\otimes\mathcal{N}_{A\rightarrow BC}^{\operatorname{H}\otimes n}%
(\varsigma_{A^{n}}^{m_{1},m_{2},t})\otimes p(t|m_{1},m_{2})|t\rangle\langle
t|_{T}%
\end{multline}
For $i\in\{1,\ldots,n\}$, let%
\begin{align}
\zeta_{M_{1}M_{2}B_{i}C_{i}Y^{i-1}T}^{i} &  \equiv(\operatorname{id}%
_{M_{1}M_{2}B_{i}C_{i}T}\otimes\mathcal{M}_{B\rightarrow Y}^{\otimes\left(
i-1\right)  }\otimes\operatorname{Tr}_{C^{i-1}})(\omega_{M_{1}M_{2}B^{i}%
C^{i}T})\\
&  =\frac{1}{M_{B}M_{C}}\sum_{m_{1},m_{2},t}|m_{1}\rangle\langle m_{1}%
|_{M_{1}}\otimes|m_{2}\rangle\langle m_{2}|_{M_{2}}\otimes\nonumber\\
&  \qquad\qquad(\mathcal{N}_{A_{i}\rightarrow B_{i}C_{i}}^{\operatorname{H}%
}\otimes\left[  \left(  \mathcal{M}_{B\rightarrow Y}^{\otimes\left(
i-1\right)  }\otimes\operatorname{Tr}_{C^{i-1}}\right)  \circ\mathcal{N}%
_{A\rightarrow BC}^{\operatorname{H}\otimes(i-1)}\right]  )(\varsigma_{A^{i}%
}^{m_{1},m_{2},t})\nonumber\\
&  \qquad\qquad\otimes p(t|m_{1},m_{2})|t\rangle\langle t|_{T}\\
&  =\frac{1}{M_{B}M_{C}}\sum_{m_{1},m_{2},y^{i-1},t}|m_{1}\rangle\langle
m_{1}|_{M_{1}}\otimes|m_{2}\rangle\langle m_{2}|_{M_{2}}\otimes\mathcal{N}%
_{A_{i}\rightarrow B_{i}C_{i}}^{\operatorname{H}}(\tau_{A_{i}}^{m_{1}%
,m_{2},t,y^{i-1}})\nonumber\\
&  \qquad\qquad\otimes p(y^{i-1}|m_{1},m_{2},t)|y^{i-1}\rangle\langle
y^{i-1}|_{Y^{i-1}}\otimes p(t|m_{1},m_{2})|t\rangle\langle t|_{T},
\end{align}
where%
\begin{equation}
\left[  \left(  \mathcal{M}_{B\rightarrow Y}^{\otimes\left(  i-1\right)
}\otimes\operatorname{Tr}_{C^{i-1}}\right)  \circ\mathcal{N}_{A\rightarrow
BC}^{\operatorname{H}\otimes(i-1)}\right]  (\varsigma_{A^{i}}^{m_{1},m_{2}%
,t})=\sum_{y^{i-1}}\tau_{A_{i}}^{m_{1},m_{2},t,y^{i-1}}\otimes p(y^{i-1}%
|m_{1},m_{2},t)|y^{i-1}\rangle\langle y^{i-1}|_{Y^{i-1}}%
\end{equation}
Taking a spectral decomposition of $\tau_{A_{i}}^{m_{1},m_{2},t,y^{i-1}}$ as
\begin{equation}
\tau_{A_{i}}^{m_{1},m_{2},t,y^{i-1}}=\sum_{s}p(s|m_{1},m_{2},t,y^{i-1}%
)\varphi_{A_{i}}^{m_{1},m_{2},t,y^{i-1},s},
\end{equation}
with each $\varphi_{A_{i}}^{m_{1},m_{2},t,y^{i-1},s}$ pure, we find that an
extension of $\zeta_{M_{1}M_{2}B_{i}C_{i}Y^{i-1}T}^{i}$ is%
\begin{multline}
\zeta_{M_{1}M_{2}B_{i}C_{i}Y^{i-1}TS}^{i}\equiv\frac{1}{M_{B}M_{C}}\sum
_{m_{1},m_{2},t,y^{i-1},s}|m_{1}\rangle\langle m_{1}|_{M_{1}}\otimes
|m_{2}\rangle\langle m_{2}|_{M_{2}}\otimes\mathcal{N}_{A_{i}\rightarrow
B_{i}C_{i}}^{\operatorname{H}}(\varphi_{A_{i}}^{m_{1},m_{2},t,y^{i-1},s})\\
\otimes p(y^{i-1}|m_{1},m_{2},t)|y^{i-1}\rangle\langle y^{i-1}|_{Y^{i-1}%
}\otimes p(t|m_{1},m_{2})|t\rangle\langle t|_{T}\otimes p(s|m_{1}%
,m_{2},t,y^{i-1})|s\rangle\langle s|_{S}.
\end{multline}
Let $\zeta_{QM_{1}M_{2}BC\overline{Y}TS}$ denote the following state:%
\begin{equation}
\zeta_{QM_{1}M_{2}BC\overline{Y}TS}\equiv\frac{1}{n}\sum_{i=1}^{n}%
|i\rangle\langle i|_{Q}\otimes\zeta_{M_{1}M_{2}B_{i}C_{i}Y^{i-1}TS}%
^{i},\label{eq:collected-state}%
\end{equation}
where $\overline{Y}$ is large enough to hold the values in each $Y^{i-1}$ (and
zero-padded if need be). Let $\zeta_{Q\overline{Q}M_{1}M_{2}\overline{M_{2}%
}BC\overline{Y}\overline{\overline{Y}}TS}$ denote an extension of
$\zeta_{QM_{1}M_{2}BC\overline{Y}TS}$, such that systems $\overline{Q}$,
$\overline{M_{2}}$, and $\overline{\overline{Y}}$ contain a classical copy of
the value in $Q$, $M_{2}$, and $\overline{Y}$, respectively.

We begin our analysis using information inequalities. Consider that%
\begin{align}
\log M_{C}  &  =I(M_{2};M_{2}^{\prime})_{\overline{\Phi}}%
\label{eq:first-C-bound}\\
&  \leq I(M_{2};M_{2}^{\prime})_{\omega}+\varepsilon\log M_{C}+h_{2}%
(\varepsilon),
\end{align}
where the inequality follows from a uniform bound for continuity of entropy
\cite{Z07,A07} (see also \cite{W15book}) and $h_{2}(\varepsilon)$ denotes the
binary entropy. Continuing, we find that%
\begin{align}
I(M_{2};M_{2}^{\prime})_{\omega}  &  \leq I(M_{2};C^{n})_{\omega}\\
&  =H(C^{n})_{\omega}-H(C^{n}|M_{2})_{\omega}\\
&  =\sum_{i=1}^{n}H(C_{i}|C^{i-1})_{\omega}-H(C_{i}|C^{i-1}M_{2})_{\omega}\\
&  \leq\sum_{i=1}^{n}H(C_{i})_{\omega}-H(C_{i}|C^{i-1}M_{2})_{\omega}
\label{eq:pf-1st-block}%
\end{align}
The first inequality follows from quantum data processing. The first equality
is an expansion of the mutual information, and the second equality is an
application of the chain rule for conditional entropy. The last inequality
follows from the fact that conditioning does not increase entropy. Continuing,%
\begin{align}
\eqref{eq:pf-1st-block}  &  \leq\sum_{i=1}^{n}H(C_{i})_{\zeta^{i}}%
-H(C_{i}|Y^{i-1}M_{2})_{\zeta^{i}}\\
&  =\sum_{i=1}^{n}I(Y^{i-1}M_{2};C_{i})_{\zeta^{i}}\\
&  =nI(\overline{Y}M_{2};C|Q)_{\zeta}\\
&  \leq nI(\overline{Y}M_{2}Q;C)_{\zeta} . \label{eq:last-C-bound}%
\end{align}
The first inequality follows due to the structure of the Hadamard broadcast
channel:\ the systems $C^{i-1}$ can be simulated from classical systems
$Y^{i-1}$, which in turn can be simulated from the systems $B^{i-1}$. The
first equality follows from the definition of mutual information. The second
equality follows by using the definition of the state in
\eqref{eq:collected-state} and the fact that conditioning on a classical
system leads to a convex combination of mutual informations. The last
inequality follows because $I(\overline{Y}M_{2};C|Q)_{\zeta}=I(\overline
{Y}M_{2}Q;C)_{\zeta}-I(Q;C)_{\zeta}\leq I(\overline{Y}M_{2}Q;C)_{\zeta}$.

We now handle the other rate bound. Consider that%
\begin{align}
\log M_{B} &  =I(M_{1};M_{1}^{\prime})_{\overline{\Phi}}\\
&  \leq I(M_{1};M_{1}^{\prime})_{\omega}+\varepsilon\log M_{B}+h_{2}%
(\varepsilon),
\end{align}
where the inequality follows from a uniform bound for continuity of entropy
\cite{Z07,A07} (see also \cite{W15book}). Continuing, we find that%
\begin{align}
I(M_{1};M_{1}^{\prime})_{\omega} &  \leq I(M_{1};B^{n}M_{2})_{\omega}\\
&  =I(M_{1};B^{n}|M_{2})_{\omega}\\
&  \leq I(M_{1}T;B^{n}|M_{2})_{\omega}\\
&  =H(B^{n}|M_{2})_{\omega}-H(B^{n}|M_{2}M_{1}T)_{\omega}\\
&  =H(B^{n}|M_{2})_{\omega}-H(C^{n}|M_{2}M_{1}T)_{\omega}\\
&  =\sum_{i=1}^{n}H(B_{i}|B^{i-1}M_{2})_{\omega}-H(C_{i}|C^{i-1}M_{2}%
M_{1}T)_{\omega}\label{eq:pf-2nd-block}%
\end{align}
The first inequality follows from quantum data processing. The first equality
follows from the chain rule for mutual information and the fact that
$I(M_{1};M_{2})_{\omega}=0$. The second inequality follows from quantum data
processing for the conditional mutual information. The second equality is an
expansion of the conditional mutual information. The third equality follows
because the state of systems $B^{n}C^{n}$ is pure when conditioned on
classical systems $M_{1}$, $M_{2}$, and $T$. The last equality applies the
chain rule for conditional entropy. Continuing,%
\begin{align}
\eqref{eq:pf-2nd-block} &  \leq\sum_{i=1}^{n}H(B_{i}|Y^{i-1}M_{2})_{\zeta^{i}%
}-H(C_{i}|Y^{i-1}M_{2}M_{1}T)_{\zeta^{i}}\\
&  \leq\sum_{i=1}^{n}H(B_{i}|Y^{i-1}M_{2})_{\zeta^{i}}-H(C_{i}|Y^{i-1}%
M_{2}M_{1}TS)_{\zeta^{i}}\\
&  =\sum_{i=1}^{n}H(B_{i}|Y^{i-1}M_{2})_{\zeta^{i}}-H(B_{i}|Y^{i-1}M_{2}%
M_{1}TS)_{\zeta^{i}}\\
&  =\sum_{i=1}^{n}I(M_{1}TS;B_{i}|Y^{i-1}M_{2})_{\zeta^{i}}\\
&  =nI(M_{1}TS;B|\overline{Y}M_{2}Q)_{\zeta}\\
&  \leq nI(M_{1}TS\overline{Q}\overline{M_{2}}\overline{\overline{Y}%
};B|\overline{Y}M_{2}Q)_{\zeta}.
\end{align}
The first inequality applies the data processing inequality for conditional
entropy:\ the $Y^{i-1}$ systems result from measurements of the $B^{i-1}$
systems, and the $C^{i-1}$ systems can be simulated by preparation channels
acting on the $Y^{i-1}$ systems. The second inequality again follows from the
data processing inequality for conditional entropy. The first equality follows
because the state of the $B_{i}C_{i}$ systems is pure when conditioned on
systems $Y^{i-1}$, $M_{2}$, $M_{1}$, $T$, and $S$. The second equality follows
from the definition of conditional mutual information. The third equality
follows by introducing the $Q$ system and evaluating the conditional mutual
information of the state $\zeta_{QM_{1}M_{2}BC\overline{Y}TS}$. The final
inequality follows from data processing for the conditional mutual information.

Putting everything together, we find that the following inequalities hold%
\begin{align}
\frac{1-\varepsilon}{n}\log M_{B}  &  \leq I(M_{1}TS\overline{Q}%
\overline{M_{2}}\overline{\overline{Y}};B|\overline{Y}M_{2}Q)_{\zeta}+\frac
{1}{n}h_{2}(\varepsilon),\\
\frac{1-\varepsilon}{n}\log M_{C}  &  \leq I(\overline{Y}M_{2}Q;C)_{\zeta
}+\frac{1}{n}h_{2}(\varepsilon).
\end{align}
Now identifying the systems $M_{1}TS\overline{Q}\overline{M_{2}}%
\overline{\overline{Y}}$ with system $Z$ in \eqref{eq:code-state}, systems
$\overline{Y}M_{2}Q$ with system $W$ in \eqref{eq:code-state}, and the state
$\varphi_{A_{i}}^{m_{1},m_{2},t,y^{i-1},s}$ with $\varphi_{A}^{z}$ in
\eqref{eq:code-state}, we can rewrite the above inequalities as follows:%
\begin{align}
\frac{1-\varepsilon}{n}\log M_{B}  &  \leq I(Z;B|W)_{\zeta}+\frac{1}{n}%
h_{2}(\varepsilon),\\
\frac{1-\varepsilon}{n}\log M_{C}  &  \leq I(W;C)_{\zeta}+\frac{1}{n}%
h_{2}(\varepsilon).
\end{align}
Now that we have established that these inequalities hold for an arbitrary
$(n,M_{B},M_{C},\varepsilon)$ code, considering a sequence $\{(n,M_{B}%
,M_{C},\varepsilon_{n})\}_{n}$ of them with $\varepsilon_{n}\rightarrow0$ as
$n\rightarrow\infty$, we find that the rate region is characterized by \eqref{eq:HB-ineq-1}--\eqref{eq:HB-ineq-2}.
\end{proof}

\section{Classical--quantum capacity of a Hadamard broadcast
channel\label{sec:CQ-task}}

\subsection{Definition of the classical--quantum capacity region of a quantum
broadcast channel}

We now recall the definition of the classical--quantum capacity region of a
quantum broadcast channel \cite{YHD2006}. Let $\mathcal{N}_{A\rightarrow BC}$
denote a quantum broadcast channel from a sender Alice to receivers Bob and
Charlie. Let $n\in\mathbb{N}$, $M_{B},M_{C}\in\mathbb{N}$, and $\varepsilon
\in\lbrack0,1]$. An $(n,M_{B},M_{C},\varepsilon)$ code for classical--quantum
communication over the broadcast channel $\mathcal{N}_{A\rightarrow BC}$
consists of quantum codewords $\left\{  \rho_{RA^{n}}^{m}\right\}  _{m}$, such
that $\dim(\mathcal{H}_{R})=M_{B}$, a decoding channel $\mathcal{D}%
_{B^{n}\rightarrow\widehat{R}}$, and a decoding POVM $\{\Gamma_{C^{n}}%
^{m}\}_{m}$. Let the state after the channel acts be as follows:%
\begin{equation}
\omega_{MRB^{n}C^{n}}\equiv\frac{1}{M_{B}}\sum_{m}|m\rangle\langle
m|_{M}\otimes\mathcal{N}_{A\rightarrow BC}^{\otimes n}(\rho_{RA^{n}}%
^{m}),\label{eq:cq-code-state}%
\end{equation}
and let the state after the decoders act be as follows:%
\begin{equation}
\omega_{MM^{\prime}R\widehat{R}}\equiv\sum_{m^{\prime}}|m^{\prime}%
\rangle\langle m^{\prime}|_{M^{\prime}}\otimes\operatorname{Tr}_{C^{n}%
}\{\Gamma_{C^{n}}^{m^{\prime}}\left[  \mathcal{D}_{B^{n}\rightarrow\widehat
{R}}(\omega_{MRB^{n}C^{n}})\right]  \}.
\end{equation}
For an $(n,M_{B},M_{C},\varepsilon)$ code, the following condition holds%
\begin{equation}
\frac{1}{2}\left\Vert \overline{\Phi}_{MM^{\prime}}\otimes\Phi_{R\widehat{R}%
}-\omega_{MM^{\prime}R\widehat{R}}\right\Vert _{1}\leq\varepsilon,
\end{equation}
where $\Phi_{R\widehat{R}}$ denotes a maximally entangled state.

A rate pair $(Q_{B},R_{C})$ is achievable for classical--quantum communication
on $\mathcal{N}_{A\rightarrow BC}$\ if for all $\varepsilon\in(0,1)$,
$\delta>0$, and sufficiently large $n$, there exists an $(n,2^{n\left[
Q_{B}-\delta\right]  },2^{n\left[  R_{C}-\delta\right]  },\varepsilon)$ code
of the above form. The classical--quantum capacity region of $\mathcal{N}%
_{A\rightarrow BC}$ is equal to the closure of all achievable rate pairs.

\subsection{Achievable rate region for an arbitrary quantum broadcast channel}

From \cite{YHD2006}, we know the following achievability statement:

\begin{theorem}
[\cite{YHD2006}]\label{thm:CQ-achieve}Given a quantum broadcast channel
$\mathcal{N}_{A\rightarrow BC}$, a rate pair $(Q_{B},R_{C})$ is achievable for
classical--quantum communication on $\mathcal{N}_{A\rightarrow BC}$ if%
\begin{align}
Q_{B}  &  \leq I(R\rangle BW)_{\theta},\\
R_{C}  &  \leq\min\{I(W;B)_{\theta},I(W;C)_{\theta}\},
\end{align}
where the information quantities are evaluated with respect to a state
$\theta_{WRBC}$ of the following form:%
\begin{equation}
\theta_{WRBC}\equiv\sum_{w}p_{W}(w)|w\rangle\langle w|_{W}\otimes
\mathcal{N}_{A\rightarrow BC}(\varphi_{RA}^{w}),
\end{equation}
with $p_{W}$ a probability distribution and $\{\varphi_{RA}^{w}\}_{w}$ a set
of pure states.
\end{theorem}

\subsection{Classical--quantum capacity region for Hadamard broadcast
channels}

\begin{theorem}
\label{thm:CQ-HBC} The classical--quantum capacity region of a Hadamard
broadcast channel $\mathcal{N}_{A\rightarrow BC}^{\operatorname{H}}$ is the
set of rate pairs $(Q_{B},R_{C})$ such that%
\begin{align}
Q_{B}  &  \leq I(R\rangle BW)_{\theta},\label{eq:cq-bnd-1}\\
R_{C}  &  \leq I(W;C)_{\theta} , \label{eq:cq-bnd-2}%
\end{align}
for some state%
\begin{equation}
\theta_{WRA}\equiv\sum_{w}p_{W}(w)|w\rangle\langle w|_{W}\otimes\varphi
_{RA}^{w},
\end{equation}
where $p_{W}$ is a probability distribution, each $\varphi_{RA}^{w}$ is a pure
state, and the information quantities are evaluated with respect to the state
$\theta_{WRBC}=\mathcal{N}_{A\rightarrow BC}^{\operatorname{H}}(\theta_{WRA})$.
\end{theorem}

\begin{proof}
The achievability part follows as a direct consequence of
Theorem~\ref{thm:CQ-achieve}, by combining data processing with the fact that
a Hadamard broadcast channel is degradable.

To begin the proof of the converse part, let a spectral decomposition for the
state $\rho_{RA^{n}}^{m}$ be as follows:%
\begin{equation}
\rho_{RA^{n}}^{m}=\sum_{t}p(t|m)\varsigma_{RA^{n}}^{m,t},
\end{equation}
with each $\varsigma_{RA^{n}}^{m,t}$ pure, so that an extension of the state
$\omega_{MRB^{n}C^{n}}$ in \eqref{eq:cq-code-state} is as follows:%
\begin{equation}
\omega_{MRB^{n}C^{n}T}=\frac{1}{M_{B}}\sum_{m,t}|m\rangle\langle m|_{M}%
\otimes\mathcal{N}_{A\rightarrow BC}^{\operatorname{H}\otimes n}%
(\varsigma_{RA^{n}}^{m,t})\otimes p(t|m)|t\rangle\langle t|_{T}.
\end{equation}
For $i\in\{1,\ldots,n\}$, let%
\begin{align}
\zeta_{MRB_{i}C_{i}Y^{i-1}T}^{i} &  \equiv(\operatorname{id}_{MRB_{i}C_{i}%
T}\otimes\mathcal{M}_{B\rightarrow Y}^{\otimes\left(  i-1\right)  }%
\otimes\operatorname{Tr}_{C^{i-1}})(\omega_{MRB^{i}C^{i}T})\\
&  =\frac{1}{M_{B}}\sum_{m,t}|m\rangle\langle m|_{M}\otimes(\mathcal{N}%
_{A_{i}\rightarrow B_{i}C_{i}}^{\operatorname{H}}\otimes\left[  \left(
\mathcal{M}_{B\rightarrow Y}^{\otimes\left(  i-1\right)  }\otimes
\operatorname{Tr}_{C^{i-1}}\right)  \circ\mathcal{N}_{A\rightarrow
BC}^{\operatorname{H}\otimes i-1}\right]  )(\varsigma_{RA^{i}}^{m,t}%
)\nonumber\\
&  \qquad\qquad\otimes p(t|m)|t\rangle\langle t|_{T}\\
&  =\frac{1}{M_{B}}\sum_{m,y^{i-1},t}|m\rangle\langle m|_{M}\otimes
\mathcal{N}_{A_{i}\rightarrow B_{i}C_{i}}^{\operatorname{H}}(\tau_{RA_{i}%
}^{m,t,y^{i-1}})\otimes p(y^{i-1}|m,t)|y^{i-1}\rangle\langle y^{i-1}%
|_{Y^{i-1}}\\
&  \qquad\qquad\otimes p(t|m)|t\rangle\langle t|_{T},
\end{align}
where%
\begin{equation}
\left[  \left(  \mathcal{M}_{B\rightarrow Y}^{\otimes\left(  i-1\right)
}\otimes\operatorname{Tr}_{C^{i-1}}\right)  \circ\mathcal{N}_{A\rightarrow
BC}^{\operatorname{H}\otimes i-1}\right]  (\varsigma_{RA^{i}}^{m,t}%
)=\sum_{y^{i-1}}\tau_{RA_{i}}^{m,t,y^{i-1}}\otimes p(y^{i-1}|m,t)|y^{i-1}%
\rangle\langle y^{i-1}|_{Y^{i-1}}.
\end{equation}
Letting $\varphi_{SRA_{i}}^{m,t,y^{i-1}}$\ denote a purification of
$\tau_{RA_{i}}^{m,t,y^{i-1}}$, we find that an extension of $\zeta
_{MRB_{i}C_{i}Y^{i-1}T}^{i}$ is%
\begin{multline}
\zeta_{MSRB_{i}C_{i}Y^{i-1}T}^{i}\equiv\frac{1}{M_{B}}\sum_{m,t,y^{i-1}%
,s}|m\rangle\langle m|_{M}\otimes\mathcal{N}_{A_{i}\rightarrow B_{i}C_{i}%
}^{\operatorname{H}}(\varphi_{SRA_{i}}^{m,t,y^{i-1}})\\
\otimes p(y^{i-1}|m,t)|y^{i-1}\rangle\langle y^{i-1}|_{Y^{i-1}}\otimes
p(t|m)|t\rangle\langle t|_{T}.
\end{multline}
Let $\zeta_{QMSRBC\overline{Y}T}$ denote the following state:%
\begin{equation}
\zeta_{QMSRBC\overline{Y}T}\equiv\frac{1}{n}\sum_{i=1}^{n}|i\rangle\langle
i|_{Q}\otimes\zeta_{MSRB_{i}C_{i}Y^{i-1}T}^{i},
\end{equation}
where $\overline{Y}$ is large enough to hold the values in each $Y^{i-1}$ (and
zero-padded if need be). 

The following information bound is a consequence of reasoning identical to
that in \eqref{eq:first-C-bound}--\eqref{eq:last-C-bound}:%
\begin{equation}
\frac{1-\varepsilon}{n}\log M_{C}\leq I(W;C)_{\zeta}+\frac{1}{n}%
h_{2}(\varepsilon),
\end{equation}
identifying $W$ as $\overline{Y}MTQ$. (To get this bound, we require a final
step of data processing to have system $T$ be included in $W$.)

We now prove the other bound. Consider that%
\begin{align}
\log M_{B}  &  =I(R\rangle\widehat{R})_{\Phi}\\
&  \leq I(R\rangle\widehat{R})_{\omega}+2\varepsilon\log M_{B}+g(\varepsilon),
\end{align}
where the inequality follows from the main result of \cite{Winter15} and
$g(\varepsilon)\equiv(1+\varepsilon)\log_{2}(1+\varepsilon)-\varepsilon
\log_{2}(\varepsilon)$, with the property that $\lim_{\varepsilon\rightarrow
0}g(\varepsilon)=0$. Continuing, we find that%
\begin{align}
I(R\rangle\widehat{R})_{\omega}  &  \leq I(R\rangle B^{n}MT)_{\omega}\\
&  =H(B^{n}|MT)_{\omega}-H(RB^{n}|MT)_{\omega}\\
&  =H(B^{n}|MT)_{\omega}-H(C^{n}|MT)_{\omega}\\
&  =\sum_{i=1}^{n}H(B_{i}|B^{i-1}MT)_{\omega}-H(C_{i}|C^{i-1}MT)_{\omega}
\label{eq:last-coh-info-block}%
\end{align}
The first inequality follows from quantum data processing. The first equality
follows from definitions. The second equality follows because the state of
systems $RB^{n}C^{n}$ is pure when conditioned on systems $MT$. The third
equality follows from the chain rule for conditional entropy. Continuing,
\begin{align}
\eqref{eq:last-coh-info-block}  &  \leq\sum_{i=1}^{n}H(B_{i}|Y^{i-1}%
MT)_{\zeta^{i}}-H(C_{i}|Y^{i-1}MT)_{\zeta^{i}}\\
&  =\sum_{i=1}^{n}H(B_{i}|Y^{i-1}MT)_{\zeta^{i}}-H(SRB_{i}|Y^{i-1}%
MT)_{\zeta^{i}}\\
&  =\sum_{i=1}^{n}I(SR\rangle B_{i}Y^{i-1}MT)_{\zeta^{i}}\\
&  =nI(SR\rangle B\overline{Y}MTQ)_{\zeta}.
\end{align}
The inequality follows from the structure of a Hadamard channel: systems
$Y^{i-1}$ can be simulated by systems $B^{i-1}$ and systems $C^{i-1}$ can be
simulated by systems $Y^{i-1}$. The second equality follows because the state
of systems $SRB_{i}$ is pure when conditioned on systems $Y^{i-1}MT$. The
third equality is by definition, and the last follows by evaluating the
coherent information of the given state and systems. Putting everything
together leads to the following two bounds:%
\begin{align}
\frac{1-2\varepsilon}{n}\log M_{B}  &  \leq I(R\rangle BW)_{\zeta}+\frac{1}%
{n}g(\varepsilon),\\
\frac{1-\varepsilon}{n}\log M_{C}  &  \leq I(W;C)_{\zeta}+\frac{1}{n}%
h_{2}(\varepsilon),
\end{align}
relabeling $R$ as $SR$ and taking $W$ as $\overline{Y}MTQ$, as stated above.
Now that we have established that these inequalities hold for an arbitrary
$(n,M_{B},M_{C},\varepsilon)$ classical--quantum code, considering a sequence
$\{(n,M_{B},M_{C},\varepsilon_{n})\}_{n}$ of them with $\varepsilon
_{n}\rightarrow0$ as $n\rightarrow\infty$, we find that the rate region is
characterized by \eqref{eq:cq-bnd-1}--\eqref{eq:cq-bnd-2}.
\end{proof}

\section{Partially entanglement-assisted classical capacity region of a
Hadamard broadcast channel\label{sec:EAC-task}}

\subsection{Definition of the partially entanglement-assisted classical
capacity region of a broadcast channel}

We now define the partially entanglement-assisted classical capacity region of
a quantum broadcast channel \cite{YHD2006}. Let $\mathcal{N}_{A\rightarrow
BC}$ denote a quantum broadcast channel from a sender Alice to receivers Bob
and Charlie. Let $n\in\mathbb{N}$, $M_{B},M_{C}\in\mathbb{N}$, and
$\varepsilon\in\lbrack0,1]$. An $(n,M_{B},M_{C},\varepsilon)$ partially
entanglement-assisted code for classical communication over the broadcast
channel $\mathcal{N}_{A\rightarrow BC}$ consists of a shared (generally
mixed)\ entangled state $\Psi_{R_{0}R_{0}^{\prime}}$, such that Alice has
system $R_{0}^{\prime}$ and Bob has system $R_{0}$. Such a code also consists
of a set of encoding channels $\{\mathcal{E}_{R_{0}^{\prime}\rightarrow A^{n}%
}^{m_{1},m_{2}}\}_{m_{1},m_{2}}$, and POVMs $\{\Lambda_{R_{0}B^{n}}^{m_{1}%
}\}_{m_{1}}$ and $\{\Gamma_{C^{n}}^{m_{2}}\}_{m_{2}}$ such that the message
pair $(m_{1},m_{2})$ is communicated with average success probability not
smaller than $1-\varepsilon$:%
\begin{equation}
\sum_{m_{1},m_{2}}\frac{\operatorname{Tr}\{(\Lambda_{R_{0}B^{n}}^{m_{1}%
}\otimes\Gamma_{C^{n}}^{m_{2}})\mathcal{N}_{A\rightarrow BC}^{\otimes
n}(\mathcal{E}_{R_{0}^{\prime}\rightarrow A^{n}}^{m_{1},m_{2}}(\Psi
_{R_{0}R_{0}^{\prime}}))\}}{M_{B}M_{C}} \geq1-\varepsilon
.\label{eq:EAC-success-prob-broadcast}%
\end{equation}

A rate pair $(R_{B},R_{C})$ is achievable for partially entanglement-assisted
classical communication on $\mathcal{N}_{A\rightarrow BC}$\ if for all
$\varepsilon\in(0,1)$, $\delta>0$, and sufficiently large $n$, there exists an
$(n,2^{n\left[  R_{B}-\delta\right]  },2^{n\left[  R_{C}-\delta\right]
},\varepsilon)$ code of the above form. The classical capacity region of
$\mathcal{N}_{A\rightarrow BC}$ is equal to the closure of all achievable rate pairs.

\subsection{Achievable rate region for an arbitrary quantum broadcast channel}

We now argue an achievability statement based on some prior coding schemes
from \cite{Shor04,cmp2005dev,itit2008hsieh}:

\begin{theorem}
\label{thm:achieve-EAC}Given a quantum broadcast channel $\mathcal{N}%
_{A\rightarrow BC}$, a rate pair $(R_{B},R_{C})$ is achievable for partially
entanglement-assisted classical communication on $\mathcal{N}_{A\rightarrow
BC}$ if%
\begin{align}
R_{B}  &  \leq I(R;B|W)_{\theta},\\
R_{C}  &  \leq\min\{I(W;B)_{\theta},I(W;C)_{\theta}\},
\end{align}
where the information quantities are evaluated with respect to a state
$\theta_{WRBC}$ of the following form:%
\begin{equation}
\sum_{w}p_{W}(w)|w\rangle\langle w|_{W}\otimes\mathcal{N}_{A\rightarrow
BC}(\varphi_{RA}^{w}),
\end{equation}
with $p_{W}$ a probability distribution and $\{\varphi_{RA}^{w}\}_{w}$ a set
of states.
\end{theorem}

\begin{proof}
We merely sketch a proof rather than work out the details, mainly because the
ideas for it have been used in a variety of contexts. The idea is similar to
that used in trade-off coding for transmitting both classical and quantum
information over a single-sender single-receiver quantum channel (see in
particular \cite[Theorem 22.5.1]{W15book}). Fix a probability distribution
$p_{W}(w)$ and a corresponding set of pure states $\{\varphi_{RA}^{w}\}_{w}$.
Pick a typical type class $T_{t}$, meaning the set of all sequences with the
same empirical distribution $t(w)$ that deviates from the true distribution
$p_{W}(w)$ by no more than $\delta>0$. All the sequences in the same type
class are related to one another by a permutation, and all of them are
strongly typical. Now we suppose that Alice and Bob share the state%
\begin{equation}
\varphi_{R^{n}A^{n}}^{w^{n}}\equiv\varphi_{R_{1}A_{1}}^{w_{1}}\otimes
\cdots\otimes\varphi_{R_{n}A_{n}}^{w_{n}},
\end{equation}
where the sequence $w^{n}\in T_{t}$, Alice has the $A^{n}$ systems, and Bob
has the $R^{n}$ systems. The reduced state after tracing over Bob's systems is
$\varphi_{A^{n}}^{w^{n}}=\varphi_{A_{1}}^{w_{1}}\otimes\cdots\otimes
\varphi_{A_{n}}^{w_{n}}$.

The main idea is for Alice to layer the messages on top of each other as in
superposition coding. She first encodes a classical message as a permutation
of the sequence $w^{n}$ (using a constant-composition code as discussed in
\cite[Section~20.3.1]{W15book}), sending the $A^{n}$ systems of the state
$\varphi_{R^{n}A^{n}}^{w^{n}}$ through $n$ uses of the broadcast channel. Bob
and Charlie then decode, and Bob neglects his systems $R^{n}$ in this first
decoding step. It is possible for each of them to decode reliably as long as
the rate $R_{C}\leq\min\{I(W;B)_{\theta},I(W;C)_{\theta}\}$. At the same time,
Alice can encode another message intended exclusively for Bob as an
entanglement-assisted code into the states $\varphi_{R^{n}A^{n}}^{w^{n}}$.
This is possible by arranging the sequence of states $\varphi_{R^{n}A^{n}%
}^{w^{n}}$ into $\left\vert W\right\vert $ blocks of $\approx np_{W}(w)$
i.i.d.~states of the form $\varphi_{RA}^{w}$. For each block, she employs the
coding scheme of \cite{itit2008hsieh} for entanglement-assisted coding at a
rate $I(R;B)_{\mathcal{N}(\varphi^{w})}$, such that the total rate for the
message intended for Bob is $\sum_{w}p_{W}(w)I(R;B)_{\mathcal{N}(\varphi^{w}%
)}=I(R;B|W)_{\theta}$. So, in a second decoding step after Bob determines
which permutation of the sequence $w^{n}$ Alice transmitted, he can rearrange
his systems $R^{n}A^{n}$ into the above standard form to decode the message
encoded in the entanglement-assisted codes. This gives the achievable rate
region above. See the discussion in the proof of \cite[Theorem 22.5.1]%
{W15book} for more details.
\end{proof}

\subsection{Partially entanglement-assisted classical capacity of a Hadamard
broadcast channel}

We now determine the partially entanglement-assisted classical capacity region
of a Hadamard broadcast channel:

\begin{theorem}
\label{thm:EAC-HBC} The partially entanglement-assisted classical capacity
region of a Hadamard broadcast channel $\mathcal{N}_{A\rightarrow
BC}^{\operatorname{H}}$ is the set of rate pairs $(R_{B},R_{C})$ such that%
\begin{align}
R_{B}  &  \leq I(R;B|W)_{\theta},\label{eq:EAC-bnd-1}\\
R_{C}  &  \leq I(W;C)_{\theta}, \label{eq:EAC-bnd-2}%
\end{align}
for some state%
\begin{equation}
\theta_{WRA}=\sum_{w}p_{W}(w)|w\rangle\langle w|_{W}\otimes\varphi_{RA}^{w},
\label{eq:EAC-code-state}%
\end{equation}
where $p_{W}$ is a probability distribution, each $\varphi_{RA}^{w}$ is a pure
state, and the information quantities are evaluated with respect to the state
$\theta_{WRBC}=\mathcal{N}_{A\rightarrow BC}^{\operatorname{H}}(\theta_{WRA})$.
\end{theorem}

\begin{proof}
The achievability part follows from a direct application of
Theorem~\ref{thm:achieve-EAC}, by combining data processing with the fact that
a Hadamard broadcast channel is degradable.

For the converse, consider an arbitrary $(n,M_{B},M_{C},\varepsilon)$ code for
the broadcast Hadamard channel $\mathcal{N}_{A\rightarrow BC}%
^{\operatorname{H}}$. Let $\omega_{M_{1}M_{2}R_{0}B^{n}C^{n}}$ denote the
following state:%
\begin{equation}
\omega_{M_{1}M_{2}R_{0}B^{n}C^{n}}\equiv\frac{1}{M_{B}M_{C}}\sum_{m_{1},m_{2}%
}|m_{1}\rangle\langle m_{1}|_{M_{1}} \otimes|m_{2}\rangle\langle m_{2}%
|_{M_{2}}\otimes\mathcal{N}_{A\rightarrow BC}^{\operatorname{H}\otimes
n}(\mathcal{E}_{R_{0}^{\prime}\rightarrow A^{n}}^{m_{1},m_{2}}(\Psi
_{R_{0}R_{0}^{\prime}})),
\end{equation}
so that this is the state before the receivers act with their measurements.
The post-measurement state is as follows:%
\begin{multline}
\omega_{M_{1}M_{2}M_{1}^{\prime}M_{2}^{\prime}} \equiv\sum_{m_{1}%
,m_{1}^{\prime},m_{2},m_{2}^{\prime}}p(m_{1},m_{1}^{\prime},m_{2}%
,m_{2}^{\prime})|m_{1}\rangle\langle m_{1}|_{M_{1}}\\
\otimes|m_{1}^{\prime}\rangle\langle m_{1}^{\prime}|_{M_{1}^{\prime}}%
\otimes|m_{2}\rangle\langle m_{2}|_{M_{2}}\otimes|m_{2}^{\prime}\rangle\langle
m_{2}^{\prime}|_{M_{2}^{\prime}},
\end{multline}
where%
\begin{equation}
p(m_{1},m_{1}^{\prime},m_{2},m_{2}^{\prime}) =\frac{\operatorname{Tr}%
\{(\Lambda_{R_{0}B^{n}}^{m_{1}^{\prime}}\otimes\Gamma_{C^{n}}^{m_{2}^{\prime}%
})\mathcal{N}_{A\rightarrow BC}^{\operatorname{H}\otimes n}(\mathcal{E}%
_{R_{0}^{\prime}\rightarrow A^{n}}^{m_{1},m_{2}}(\Psi_{R_{0}R_{0}^{\prime}%
}))\}}{M_{B}M_{C}}.
\end{equation}
From the condition in \eqref{eq:EAC-success-prob-broadcast}, it follows that%
\begin{equation}
\frac{1}{2}\left\Vert \omega_{M_{1}M_{2}M_{1}^{\prime}M_{2}^{\prime}%
}-\overline{\Phi}_{M_{1}M_{1}^{\prime}}\otimes\overline{\Phi}_{M_{2}%
M_{2}^{\prime}}\right\Vert _{1}\leq\varepsilon,
\end{equation}
where $\overline{\Phi}_{M_{i}M_{i}^{\prime}}\equiv\frac{1}{|M_{i}|}\sum
_{m_{i}}|m_{i}\rangle\langle m_{i}|_{M_{i}}\otimes|m_{i}\rangle\langle
m_{i}|_{M_{i}^{\prime}}$ is the maximally correlated state for $i\in\{1,2\}$.

For a fixed $m_{2}$, let $\varsigma_{M_{1}R_{0}A^{n}}^{m_{2}}$ denote the
following state:%
\begin{equation}
\varsigma_{M_{1}R_{0}A^{n}}^{m_{2}}\equiv\frac{1}{M_{B}}\sum_{m_{1}}%
|m_{1}\rangle\langle m_{1}|_{M_{1}}\otimes\mathcal{E}_{R_{0}^{\prime
}\rightarrow A^{n}}^{m_{1},m_{2}}(\Psi_{R_{0}R_{0}^{\prime}}),
\end{equation}
and let $\varsigma_{TM_{1}R_{0}A^{n}}^{m_{2}}$ denote a purification of it, so
that%
\begin{equation}
\omega_{M_{1}M_{2}R_{0}TB^{n}C^{n}}\equiv\frac{1}{M_{C}}\sum_{m_{2}}%
|m_{2}\rangle\langle m_{2}|_{M_{2}}\otimes\mathcal{N}_{A\rightarrow
BC}^{\operatorname{H}\otimes n}(\varsigma_{TM_{1}R_{0}A^{n}}^{m_{2}})
\end{equation}
is an extension of $\omega_{M_{1}M_{2}R_{0}B^{n}C^{n}}$. For $i\in
\{1,\ldots,n\}$, let%
\begin{align}
&  \zeta_{M_{1}M_{2}R_{0}B_{i}C_{i}Y^{i-1}T}^{i}\nonumber\\
&  \equiv(\operatorname{id}_{M_{1}M_{2}R_{0}B_{i}C_{i}T}\otimes\mathcal{M}%
_{B\rightarrow Y}^{\otimes\left(  i-1\right)  }\otimes\operatorname{Tr}%
_{C^{i-1}})(\omega_{M_{1}M_{2}R_{0}B^{i}C^{i}T})\\
&  =\frac{1}{M_{C}}\sum_{m_{2}}|m_{2}\rangle\langle m_{2}|_{M_{2}}%
\otimes(\mathcal{N}_{A_{i}\rightarrow B_{i}C_{i}}^{\operatorname{H}}%
\otimes\left[  \left(  \mathcal{M}_{B\rightarrow Y}^{\otimes\left(
i-1\right)  }\otimes\operatorname{Tr}_{C^{i-1}}\right)  \circ\mathcal{N}%
_{A\rightarrow BC}^{\operatorname{H}\otimes i-1}\right]  )(\varsigma
_{TM_{1}R_{0}A^{i}}^{m_{2}})\\
&  =\frac{1}{M_{C}}\sum_{m_{2},y^{i-1}}|m_{2}\rangle\langle m_{2}|_{M_{2}%
}\otimes\mathcal{N}_{A_{i}\rightarrow B_{i}C_{i}}^{\operatorname{H}}%
(\tau_{TM_{1}R_{0}A_{i}}^{m_{2},y^{i-1}})\otimes p(y^{i-1}|m_{2}%
)|y^{i-1}\rangle\langle y^{i-1}|_{Y^{i-1}},
\end{align}
where%
\begin{equation}
\left[  \left(  \mathcal{M}_{B\rightarrow Y}^{\otimes\left(  i-1\right)
}\otimes\operatorname{Tr}_{C^{i-1}}\right)  \circ\mathcal{N}_{A\rightarrow
BC}^{\operatorname{H}\otimes i-1}\right]  (\varsigma_{TM_{1}R_{0}A^{i}}%
^{m_{2}})=\sum_{y^{i-1}}\tau_{TM_{1}R_{0}A_{i}}^{m_{2},y^{i-1}}\otimes
p(y^{i-1}|m_{2})|y^{i-1}\rangle\langle y^{i-1}|_{Y^{i-1}}%
\end{equation}
Taking a purification of $\tau_{TM_{1}R_{0}A_{i}}^{m_{2},y^{i-1}}$ as
$\varphi_{STM_{1}R_{0}A_{i}}^{m_{2},y^{i-1}}$, we find that an extension of
$\zeta_{M_{1}M_{2}R_{0}B_{i}C_{i}Y^{i-1}T}^{i}$ is%
\begin{multline}
\zeta_{M_{1}M_{2}R_{0}B_{i}C_{i}Y^{i-1}TS}^{i}\equiv\frac{1}{M_{C}}\sum
_{m_{2},y^{i-1}}|m_{2}\rangle\langle m_{2}|_{M_{2}}\\
\otimes\mathcal{N}_{A_{i}\rightarrow B_{i}C_{i}}^{\operatorname{H}}%
(\varphi_{STM_{1}R_{0}A_{i}}^{m_{2},y^{i-1}})\otimes p(y^{i-1}|m_{2}%
)|y^{i-1}\rangle\langle y^{i-1}|_{Y^{i-1}}.
\end{multline}
Let $\zeta_{QM_{1}M_{2}R_{0}BC\overline{Y}TS}$ denote the following state:%
\begin{equation}
\zeta_{QM_{1}M_{2}R_{0}BC\overline{Y}TS}\equiv\frac{1}{n}\sum_{i=1}%
^{n}|i\rangle\langle i|_{Q}\otimes\zeta_{M_{1}M_{2}R_{0}B_{i}C_{i}Y^{i-1}%
TS}^{i},
\end{equation}
where $\overline{Y}$ is large enough to hold the values in each $Y^{i-1}$ (and
zero-padded if need be). 

The following information bound is a consequence of reasoning identical to
that in \eqref{eq:first-C-bound}--\eqref{eq:last-C-bound}:%
\begin{equation}
\frac{1-\varepsilon}{n}\log M_{C}\leq I(W;C)_{\zeta}+\frac{1}{n}%
h_{2}(\varepsilon),
\end{equation}
identifying $W$ as $\overline{Y}M_{2}Q$.

We now handle the other rate bound. Consider that%
\begin{align}
\log M_{B} &  =I(M_{1};M_{1}^{\prime})_{\overline{\Phi}}\\
&  \leq I(M_{1};M_{1}^{\prime})_{\omega}+\varepsilon\log M_{B}+h_{2}%
(\varepsilon),
\end{align}
where the inequality follows from a uniform bound for continuity of entropy
\cite{Z07,A07} (see also \cite{W15book}). Continuing, we find that%
\begin{align}
I(M_{1};M_{1}^{\prime})_{\omega} &  \leq I(M_{1};B^{n}R_{0}M_{2})_{\omega}\\
&  =I(M_{1};B^{n}R_{0}|M_{2})_{\omega}\\
&  =I(M_{1}R_{0};B^{n}|M_{2})_{\omega}+I(M_{1};R_{0}|M_{2})_{\omega}%
-I(B^{n};R_{0}|M_{2})_{\omega}\\
&  \leq I(M_{1}R_{0};B^{n}|M_{2})_{\omega}\\
&  \leq I(M_{1}R_{0}T;B^{n}|M_{2})_{\omega}\\
&  =H(B^{n}|M_{2})_{\omega}-H(B^{n}|M_{1}R_{0}TM_{2})_{\omega}\\
&  =H(B^{n}|M_{2})_{\omega}+H(B^{n}|C^{n}M_{2})_{\omega}%
\label{eq:EAC-1st-block}%
\end{align}
The first inequality follows from quantum data processing. The first equality
follows from the chain rule for mutual information and the fact that
$I(M_{1};M_{2})_{\omega}=0$. The second equality is an identity for
conditional mutual information. The second inequality follows because
$I(M_{1};R_{0}|M_{2})_{\omega}=0$\ (i.e., the reduced state on these systems
is a product state) and $I(B^{n};R_{0}|M_{2})_{\omega}\geq0$. The third
inequality follows from quantum data processing for the conditional mutual
information. The third equality is an expansion of the conditional mutual
information. The third equality follows because the state of systems
$M_{1}R_{0}TB^{n}C^{n}$ is pure when conditioned on classical system $M_{2}$.
Continuing,%
\begin{align}
\eqref{eq:EAC-1st-block} &  =\sum_{i=1}^{n}H(B_{i}|B^{i-1}M_{2})_{\omega
}+H(B_{i}|B^{i-1}C^{n}M_{2})_{\omega}\\
&  \leq\sum_{i=1}^{n}H(B_{i}|Y^{i-1}M_{2})_{\zeta^{i}}+H(B_{i}|Y^{i-1}%
C_{i}M_{2})_{\zeta^{i}}\\
&  =\sum_{i=1}^{n}H(B_{i}|Y^{i-1}M_{2})_{\zeta^{i}}-H(B_{i}|R_{0}%
STM_{1}Y^{i-1}M_{2})_{\zeta^{i}}\\
&  =\sum_{i=1}^{n}I(R_{0}STM_{1};B_{i}|Y^{i-1}M_{2})_{\zeta^{i}}\\
&  =I(R_{0}STM_{1};B|\overline{Y}M_{2}Q)_{\zeta}.
\end{align}
The first equality follows from the chain rule for conditional entropy. The
first inequality applies the data processing inequality for conditional
entropy:\ the $Y^{i-1}$ systems result from measurements of the $B^{i-1}$
systems and then we discard the $C^{i-1}$ systems. The second equality follows
because the state of the $R_{0}STM_{1}B_{i}C_{i}$ systems is pure when
conditioned on systems $Y^{i-1}$ and $M_{2}$. The second equality follows from
the definition of conditional mutual information. The third equality follows
by introducing the $Q$ system and evaluating the conditional mutual
information of the state $\zeta_{QM_{1}M_{2}R_{0}BC\overline{Y}TS}$.

Putting everything together, we find that the following inequalities hold%
\begin{align}
\frac{1-\varepsilon}{n}\log M_{B}  &  \leq I(R_{0}STM_{1};B|\overline{Y}%
M_{2}Q)_{\zeta}+\frac{1}{n}h_{2}(\varepsilon),\\
\frac{1-\varepsilon}{n}\log M_{C}  &  \leq I(\overline{Y}M_{2}Q;C)_{\zeta
}+\frac{1}{n}h_{2}(\varepsilon).
\end{align}
Now identifying the systems $R_{0}STM_{1}$ with system $R$ in
\eqref{eq:EAC-code-state}, systems $\overline{Y}M_{2}Q$ with system $W$ in
\eqref{eq:EAC-code-state}, and the state $\varphi_{STM_{1}R_{0}A_{i}}%
^{m_{2},y^{i-1}}$ with $\varphi_{RA}^{w}$ in \eqref{eq:EAC-code-state}, we can
rewrite the above inequalities as follows:%
\begin{align}
\frac{1-\varepsilon}{n}\log M_{B}  &  \leq I(R;B|W)_{\zeta}+\frac{1}{n}%
h_{2}(\varepsilon),\\
\frac{1-\varepsilon}{n}\log M_{C}  &  \leq I(W;C)_{\zeta}+\frac{1}{n}%
h_{2}(\varepsilon).
\end{align}
Now that we have established that these inequalities hold for an arbitrary
$(n,M_{B},M_{C},\varepsilon)$ code, considering a sequence $\{(n,M_{B}%
,M_{C},\varepsilon_{n})\}_{n}$ of them with $\varepsilon_{n}\rightarrow0$ as
$n\rightarrow\infty$, we find that the rate region is characterized by \eqref{eq:EAC-bnd-1}--\eqref{eq:EAC-bnd-2}.
\end{proof}

\begin{remark}
In all of the capacity theorems established here (Theorems~\ref{thm:CC-HBC},
\ref{thm:CQ-HBC}, and \ref{thm:EAC-HBC}), the rate $R_{C}$ can be replaced by
the sum of a rate $R_{BC}$ and $R_{C}$, where $R_{BC}$ is the rate of a common
message intended for both Bob and Charlie, while $R_{C}$ is the rate of a
message intended for Charlie. This is because all of our converses go through
with this modification, and at the same time, all of the achievability parts
make use of a superposition coding technique, in which Bob first decodes the
message intended for Charlie before decoding the message intended for him.
\end{remark}

\section{Conclusion}

\label{sec:conclusion}This paper solves the classical capacity,
classical--quantum capacity, and partially entanglement-assisted classical
capacity of Hadamard broadcast channels. As such, these channels might
naturally be viewed as a quantum extension of the notion of a degraded
broadcast channel. Essential in all of our analyses is the structure of a
Hadamard broadcast channel in which the $C$ system can be simulated in two
steps:\ first a measurement channel taking the $B$ system to a classical $Y$
system and then a preparation channel taking the classical $Y$ system to the
$C$ system. This structure allows for a classical auxiliary variable to
include the classical $Y$ system in each of the problems we considered, as is
common in network classical information theory \cite{el2010lecture}.

Much remains to be understood about including fully quantum systems in
auxiliary variables, but there has been some progress on this front
\cite{BG14}\ and various information-theoretic tasks have been characterized
using auxiliary quantum variables \cite{SSW08,S08,BO12,TGW14IEEE}. However, in
many of these cases, it is not known whether a bound can be placed on the
dimension of an auxiliary quantum system and so quantities involving quantum
auxiliary variables are not known to be tractable. At the least, the structure
of a Hadamard broadcast channel allows for circumventing this problem and
yields a complete characterization of some of its capacities.

\bigskip\textbf{Acknowledgments}. We are grateful to Haoyu Qi for discussions
related to the topic of this paper.  QLW is supported by the NFSC (Grants
No.~61272057, No.~61309029 and No.~61572081) and funded by the China
Scholarship Council (Grant No.~201506470043). SD acknowledges support from the
LSU Graduate School Economic Development Assistantship. MMW acknowledges
support from the NSF under Award No.~CCF-1350397.


\end{document}